\documentclass[journal,11pt,onecolumn,draftcls]{IEEEtran}

\usepackage{amssymb,amsmath,amsfonts,amsthm,color,mathrsfs,verbatim, wrapfig}
\usepackage{inputenc, graphicx}
\usepackage{tikz,pgfplots}
\usepackage{url}
\usetikzlibrary{matrix,decorations.pathreplacing,calc,positioning,fit}

\newtheorem{theorem}{Theorem}
\newtheorem{corollary}{Corollary}
\newtheorem{lemma}{Lemma}
\newtheorem{proposition}{Proposition}
\newtheorem{definition}{Definition}
\newtheorem{remark}{Remark}
\newtheorem{example}{Example}

\newcommand{\F}{\mathbf{F}}
\newcommand{\Sym}{\mathfrak{S}}
\newcommand{\rep}{\mathrm{Rep}}
\newcommand{\lcm}{\mathrm{lcm}}
\newcommand{\Aff}{\mathrm{Aff}}
\newcommand{\supp}{\mathrm{supp}}

\newcommand{\bred}{\begin{color}{red}}
\newcommand{\ered}{\end{color}}

\begin{document}
 
\title{$t$-Private Information Retrieval Schemes Using Transitive Codes}

\author{Ragnar Freij-Hollanti\thanks{Ragnar Freij-Hollanti is at the Technical University of Munich, Munich, Germany.
Email: ragnar.freij@tum.de}, Oliver W.~Gnilke, Camilla Hollanti, \emph{Member, IEEE}\thanks{Oliver Gnilke and Camilla Hollanti are at Aalto University, Helsinki, Finland. Hollanti is also affiliated with the Technical University of Munich via a \emph{Hans Fischer} Fellowship.  Emails: \{oliver.gnilke, camilla.hollanti\}@aalto.fi}, \\ Anna-Lena Horlemann-Trautmann\thanks{Anna-Lena Horlemann-Trautmann is at the University of St. Gallen, St. Gallen, Switzerland.
Email: anna-lena.horlemann@unisg.ch}, David Karpuk\thanks{David Karpuk is at Universidad de los Andes, Bogot\'a, Colombia.
Email: da.karpuk@uniandes.edu.co}, Ivo Kubjas\thanks{Ivo Kubjas is at the University of Tartu, Tartu, Estonia.
Email: ivokub@ut.ee} \thanks{
This work is supported in part by the Academy of Finland, under grants \#276031, \#282938, and \#303819, and by the Technical University of Munich -- Institute for Advanced Study, funded by the German Excellence Initiative and the European Union Seventh Framework Programme under grant agreement \#291763. }}

\maketitle


\begin{abstract}
This paper presents private information retrieval (PIR) schemes for coded storage with colluding servers, which are not restricted to maximum distance separable (MDS) codes. PIR schemes for general linear codes are constructed and the resulting PIR rate is calculated explicitly.  It is shown that codes with transitive automorphism groups yield the highest possible rates obtainable with the proposed scheme.  This rate coincides with the known asymptotic PIR capacity for MDS-coded storage systems without collusion. While many PIR schemes in the literature require field sizes that grow with the number of servers and files in the system, we focus especially on the case of a binary base field, for which Reed-Muller codes serve as an important and explicit class of examples.  
\end{abstract}

\section{Introduction}

Private information retrieval (PIR) seeks to retrieve data from a database without disclosing information
about the identity of the data items retrieved, and was introduced by Chor \emph{et al.} in
\cite{chor1995private}, \cite{chor1998private}. The classic PIR model of \cite{chor1998private}
views the database as an $m$-bit binary string $ x = (x^1, \dots, x^m)
\in \{0,1\}^m$, and assumes that the user wants to retrieve a single bit
$x^i$ without revealing any information about the index $i$. The \emph{PIR rate}, or simply \emph{rate}, of a PIR scheme is measured as the ratio of the gained information over the downloaded information, while upload costs of the requests are usually ignored.  The trivial solution of downloading the entire database is the only way to guarantee \emph{information-theoretic privacy} in the case of a single server~\cite{chor1998private}, but replicating the database onto $k$ servers that do not communicate can significantly increase the rate, as in \cite{chor1998private}, \cite{efremenko20093}, \cite{dvir20142} and the references therein.

Shah \emph{et al.}  recently introduced a model of PIR for coded data~\cite{shah2012}, \cite{shah2014}. Here, all files are distributed over the servers according to a storage code. It is shown in~\cite{shah2014} that for a suitably constructed storage code, privacy can be guaranteed by downloading a single bit more than the size of the desired file. However, this requires exponentially many servers in terms of the number of files. Blackburn \emph{et al.}  achieved the same low download complexity with a linear number of servers~\cite{blackburn2016small}, but this is still far from applicable storage systems where the number of files tends to dwarf the number of servers.


Modern distributed storage systems require communication between servers to recover data in the case of node failure.  As such, it is natural in a PIR scheme to allow the servers to \emph{collude}, that is, to assume the servers inform each other of their interaction with the user.  Explicit PIR schemes for coded storage and colluding servers were previously considered in \cite{razan_salim}, \cite{FGHK16}, and \cite{patternISIT2017}.

 
 The maximum possible rate, or \emph{capacity} of a PIR scheme for a replicated storage system was derived in~\cite{sun_jafar_1} without collusion and in~\cite{sun_jafar_2} with collusion. The corresponding PIR capacity of an MDS-coded storage system was given in~\cite{bananaman}, in the case of no colluding servers.  The PIR capacity of MDS-coded storage systems with colluding servers is only known for some particular sets of parameters \cite{sun_jafar_conjecture}. In \cite{norwegians} PIR from non-MDS coded non-colluding storage systems is considered and in some examples a PIR rate equal to that of MDS coded systems of the same code rate is achieved. To the best of our knowledge, PIR capacity expressions for non-MDS coded storage have not been discussed in the literature.
 
 \subsection{Contributions and Related Work}
While explicit PIR schemes which achieve capacity are constructed in \cite{sun_jafar_1}, \cite{sun_jafar_2}, and \cite{bananaman}, they require the base field to be large.  If $n$ is the number of servers and $M$ is the number of files, the capacity-achieving schemes of \cite{sun_jafar_2} require a field size of $q = O(n^{M})$, since they rely on the existence of MDS codes of high lengths.  Realistic storage systems, however, may operate over fields of small size to keep the complexity of the involved operations manageable.  One would naturally then like to construct explicit PIR schemes over small base fields.

In this work  we construct PIR schemes based on general linear codes, and concentrate in particular on binary Reed--Muller (RM) codes.  The schemes described in \cite{FGHK16} employed Generalized Reed-Solomon (GRS) codes, and the resulting analysis of the achievable rate relied on the \emph{star product} $C\star D$ of two GRS codes $C$ and $D$ again being a GRS code. The class of RM codes is closed under the star product operation as well, and thus naturally lends itself to be employed using the PIR scheme of \cite{FGHK16}.  However, RM codes have the advantage of being defined over the binary field $\F_2$. When comparing GRS and RM codes of equal length and dimension, it is shown here that the same PIR rates as with GRS codes can be achieved in the non-colluding case. For a fixed PIR rate, however, RM codes provide less protection against collusion due to their lower minimum distance. Nevertheless, it is shown that the $t$-PIR RM schemes presented here still provide protection against a substantial fraction of colluding sets of sizes slightly bigger than $t$. 

In more detail, the main contributions of this paper are:

\begin{itemize}
\item Given an arbitrary storage code $C$ and retrieval code $D$, we construct a PIR scheme with rate $(d_{C\star D}-1)/n$ which protects against $(d_{D^\perp}-1)$-collusion, where $n$ is the length of $C$ and $D$ (and equal to the number of servers in the system).  
\item For some classes of $C$ and $D$, and in particular when $C$ and $C\star D$ have transitive automorphism groups, we  improve the above construction to one which achieves a PIR rate of $\dim(C\star D)^\perp/n$, the maximum possible for the presented scheme.  This also coincides with the asymptotic PIR capacity in the MDS-coded non-colluding case (in which $D$ is the repetition code).  
\item We apply our construction to the case when $C$ and $D$ are binary Reed--Muller codes, resulting in a large class of PIR schemes defined over $\F_2$ for coded storage systems with colluding servers.  
\item As a corollary of these results, we also improve on the PIR rates of some of the distributed storage systems studied in \cite{norwegians}.

\end{itemize}

 RM codes have previously been considered for PIR in other settings~\cite{vajha2017}. There, the system model is different from the present paper, in that the coding is between different files and the primary goal is to minimize storage overhead for a given PIR scheme, along the same lines as in \cite{favaya2015}. In our work, coding is between different blocks of the same file, and the goal is to minimize the download overhead for fixed storage codes.

\subsection{Paper Organization} 
 	
The rest of this paper is organized as follows. In Section \ref{sec:intro}, we introduce the standard system model of PIR for coded storage, including the notion of server collusion. Sections \ref{sec:2A} and \ref{sec:2B} recall the star product scheme from \cite{FGHK16}, and it is shown that we can always achieve a PIR rate of $(d_{C\star D}-1)/n$ while protecting against $(d_{D^\perp}-1)$-collusion, for any linear codes $C$ and $D$. In Section \ref{sec:improve}, we increase the rate of our star product schemes via a careful study of the combinatorics of information sets of the storage and retrieval codes. As a corollary of this, we reproduce and improve upon results of \cite{norwegians} for some specific storage codes. In Section \ref{sec:trans}, we show that when the codes $C$ and $C\star D$ have transitive automorphism groups, then the PIR rate can be further increased from  $(d_{C\star D}-1)/n$ to $\dim(C\star D)^\perp/n$. Section \ref{sec:3} instantiates our results in the case when the storage and retrieval codes are both binary Reed Muller codes. For such codes, explicit trade-offs between storage rate and PIR rate are derived. Section \ref{sec:conclusions} concludes the paper. 
	
\subsection{Introduction to Private Information Retrieval from Coded Storage}\label{sec:intro}

Let us describe the distributed storage systems we consider; this setup follows that of \cite{razan_salim,FGHK16,bananaman}.  To provide clear and concise notation, we have consistently used superscripts to refer to files or parts of files, superscripts in parenthesis to refer to iterations of an algorithm, subscripts to refer to servers, and parenthetical indices for entries of a vector.  %
Hence, for example, the query $q^{[w](\gamma)}_j$ is sent to the $j^{th}$ server in the $\gamma^{th}$ iteration when downloading file $w$.  In general, we distinguish the file we wish to download from an arbitrary file in the system by using $x^{[w]}$ for the former, and $x^i$ for the latter.  We denote by $\F$ an arbitrary finite field, of unspecified size except when explicitly stated.

Suppose we have files $x^1,\ldots, x^M\in \F^{b\times k}$.  The considered data storage scheme proceeds by arranging the files into a $bM \times k$ matrix
\begin{equation} \
X=
\begin{bmatrix}
x^1 \\ \vdots \\ x^M
\end{bmatrix}.
\end{equation}
Each file $x^i$ is encoded using a linear $[n,k,d_C]_q$-code $C$ having generator matrix $G_C$, into an encoded file $ y^i =  x^iG_C\in \F^{b\times n}$.  In matrix form, we encode the matrix $X$ into a matrix $Y$ by right-multiplying by $G_C$:
\begin{equation} Y= X G_C = 
\begin{bmatrix}
y^1 \\ \vdots \\ y^M
\end{bmatrix}
= 
\begin{bmatrix}
 y_1 & \cdots &  y_n
\end{bmatrix}\,.
\end{equation}
The $j^{th}$ column $ y_j\in\F^{bM\times 1}$ of the matrix $Y$ is stored by the $j^{th}$ server. Such a storage system can tolerate up to any $d_C-1$ servers failing.  If $C$ is an MDS code, the resulting distributed storage system is maximally robust against server failure.

The following defines precisely what we mean by a PIR scheme; for simplicity we have limited ourselves to simple linear schemes, which suffices to describe all the schemes constructed in this paper.  For convenience, we use the notation $[n] = \{1,\ldots,n\}$ for any positive integer $n$ throughout the paper.
	
	\begin{definition}\label{PIR_def}\cite[Definition 4]{FGHK16}
Suppose we have a distributed storage system $Y = XG_C$ as above, where $M$ files are stored across $n$ servers.  A \emph{PIR scheme} for such a storage system consists of:
\begin{itemize}
\item[1.] For each file index $w$, a probability space $(\mathcal{Q}^{[w]},\mu^{[w]})$ of \emph{queries}.  When the user wishes to download $x^{[w]}\in \F^{b\times k}$, a query $q^{[w]} \in \mathcal{Q}^{[w]}$ is selected randomly according to the probability measure $\mu^{[w]}$.  Each $q^{[w]}$ is itself a tuple $ q^{[w]} = \left( q^{[w]}_1,\ldots, q^{[w]}_n\right)$, where $ q^{[w]}_j\in \F^{1\times bM}$ is sent to the $j^{th}$ server.
\item[2.] \emph{Responses} $r^{[w]}_j =  q^{[w]}_j\cdot y_j\in \F$ which the servers compute and transmit to the user.  We set $ r^{[w]} = \left(r^{[w]}_1,\ldots, r^{[w]}_n\right)$ to be the total response vector.
\item[3.] An \emph{iteration process}, which repeats Steps 1 and 2 a total of $s$ times until the desired file $ x^{[w]}$ can be reconstructed from the $s$ responses  $r^{[w]}$.
\item[4.] A \emph{reconstruction function} which takes as input all of the total response vectors $r^{[w]}\in \F^n$ over all $s$ iterations of the scheme, and outputs the desired file $x^{[w]}$.
\end{itemize}
\end{definition}

\begin{definition}\cite[Definition 5]{FGHK16}
The \emph{PIR rate} of a PIR scheme is defined to be $\frac{bk}{ns}$.
\end{definition}

In the above, we view $b$ and $s$ as parameters that we are free to vary to enable the user to download exactly one file; see Theorem \ref{thmPIR} below and the following discussion.



\begin{definition}\label{collusion}
We call a set $T \subseteq [n]$ a \emph{colluding set} if it is possible for the servers indexed by $T$ to share their quaries in an attempt to deduce the index of the requested file.  A PIR scheme \emph{protects against the colluding set $T=\{j_1,\ldots,j_t\}\subseteq[n]$} if we have
\begin{equation}\label{mi}
I(Q^{[w]}_T;w) = 0
\end{equation}
where $Q^{[w]}_T$ denotes the joint distribution of all tuples $\{q^{[w]}_{j_1},\ldots,q^{[w]}_{j_t}\}$ of queries sent to the servers in $T$ over all $s$ iterations of the PIR scheme, and $I(\cdot\ ;\cdot)$ denotes the mutual information of two random variables.  In other words, there exists a probability distribution $(\mathcal{Q}_T,\mu_T)$ such that, for all $w\in [M]$, the projection of $(\mathcal{Q}^{[w]},\mu^{[w]})$ to the coordinates in $T$ is $(\mathcal{Q}_T,\mu_T)$.  

If a PIR scheme protects against all colluding sets $T$ of size $\leq t$, we say it \emph{protects against $t$-collusion}.
\end{definition}
	
	Stated somewhat less formally, if a PIR scheme protects against the colluding set $T$, the servers in $T$ will not learn anything about the index $w$ of the file that is being requested, even after sharing their quaries with each other.
	
	For the rest of this paper we exclusively consider linear schemes that use uniform distributions on the query spaces, as in the following fundamental example.
	
	\begin{example}
		Let $n=2$ servers each store a copy of a database $x$ consisting of $M$ files $x^i \in \F$. 
		To retrieve the $w^{th}$ file the user chooses uniformly at random an element $u\in \F^M$ and constructs the queries as $ q^{[w]}=(q^{[w]}_1,q^{[w]}_2)=(u,u+e_w)$, where $e_w$ is the $w^{th}$ standard basis vector of length $M$. The space of all queries therefore is given by $\mathcal{Q}^{[w]}=\{(u,v) : v-u = e_w\}$ and $\mu^{[w]}$ is the uniform probability measure on $\mathcal{Q}^{[w]}$. The responses $r^{[w]}_j:= q^{[w]}_j \cdot x $ are calculated as the inner product of the database with the quaries and reconstruction is achieved by subtraction of the responses, $x^{[w]}=r^{[w]}_2-r^{[w]}_1$.
		
		This scheme is secure against either server individually, as both projections of any query space $\mathcal{Q}^{[i]}$ onto a coordinate are identical to the complete ambient space $\F^n$ with uniform measure. It does not, however, protect against $2$-collusion, as the two servers can jointly observe the index $i$ by computing the difference of their query vectors.
	\end{example}


\section{A General PIR Scheme for Coded Data Stored over Colluding Servers}

In this section, we recall the star product scheme from \cite{FGHK16}, and prove that we can always achieve a PIR rate of $(d_{C\star D}-1)/n$ while protecting against $(d_{D^\perp}-1)$-collusion, for arbitrary linear codes $C$ and $D$. We then proceed to show how the rate of our star product schemes can be increased by carefully studying the combinatorics of information sets of the storage and retrieval codes. We also show that when the codes $C$ and $C\star D$ both have transitive automorphism groups, the PIR rate can be further increased from  $(d_{C\star D}-1)/n$ to $\dim(C\star D)^\perp/n$. 
	
	\subsection{Background on Star Product PIR Schemes}\label{sec:2A}
	In this section we briefly summarize the methods of \cite{FGHK16}, which construct explicit PIR schemes for coded data which protect against $t$-collusion. The crucial ingredients are the storage code $C\subseteq \F^n$, another linear code $D\subseteq\F^n$ used to construct the queries, and the star product $C\star D$, the dual of which is used for decoding. We denote by $C^\perp$ the dual code of $C$. Let us first recall the definition of the star product (also called Schur or Hadamard product) of two codes.
	
	\begin{definition}
Let $C$ and $D$ be linear codes of length $n$ over $\F$.  We define their \emph{star product} $C\star D$ to be
	\[
	C\star D = \mathrm{span}\{c\star d \mid c\in C, d\in D\}\subseteq \F^n,
	\]
where $c\star d$ is the star product of vectors, \emph{i.e.}, component-wise product of the vectors $c$ and $d$. The star product $C\star D$  is again a linear code of length $n$ by definition. 
\end{definition}

	Recall that an \emph{information set} of an $[n,k]$-code is a subset of $[n]$ of size $k$, corresponding to an invertible submatrix of the generator matrix of the code.  The main theorem of \cite{FGHK16} is the following: 

\begin{theorem}\label{thmPIR}\cite{FGHK16}
Let $C\subseteq\F^n$ be an $[n,k,d_C]$ linear storage code and let $D\subseteq \F^n$ be a linear code such that either (i) $d_{C\star D}-1\leq k$, or (ii) 
there exists $J\subseteq[n]$ of size $d_{C\star D}-1$ such that every subset of $J$ of size $k$ is an information set of $C$.  Then there exists a linear PIR scheme for the distributed storage system $Y = XG_C$ with rate $(d_{C\star D}-1)/n$ which protects against $(d_{D^\perp}-1)$-collusion. 
\end{theorem}
		
		Let us recall how one iteration of the scheme works in the simple case where each file consists of $b = 1$ row, that is, $x^i\in \F^{1\times k}$.  To privately retrieve a wanted file $x^{[w]}$, for every file $x^{i}$ in the database a codeword $d^i$ is chosen uniformly at random from the code $D\subseteq \F^n$.
		 A vector $e \in\F^n\setminus D$ is then added to $ d^w$.  The query $q_j^{[w]}\in \F^{1\times M}$ sent to the $j^{th}$ server is then
		\[
		q^{[w]}_j = (d^1(j),\ldots,d^w(j)+e(j),\ldots,d^M(j))
		\]
and the servers respond with
		\[
		\begin{split}
		\left(r^{[w]}_1,\ldots,r^{[w]}_n\right) &= \left(q^{[w]}_1\cdot y_1 ,\ldots, q^{[w]}_n\cdot y_n\right) \in C\star D + C\star e.
		\end{split}
		\]
The support of $e$ is chosen so that right-multiplying the vector $\left(r^{[w]}_1,\ldots,r^{[w]}_n\right)$ with the parity check matrix of $C\star D$ reveals $d_{C\star D}-1$ coordinates of $y^{[w]}$, coming from the $C\star e$ summand in the above expression. 

The scheme protects against $(d_{D^\perp}-1)$-collusion because every $t=d_{D^\perp}-1$ columns of the generator matrix of $D$ (which is a parity check matrix of $D^\perp$) are linearly independent, hence the joint distribution of the queries at any $t$ servers is the uniform distribution on $(\F^M)^t$.

	More generally, suppose we want to download a file $x^{[w]}\in \F^{b\times k}$.  We denote the encoded version by $x^{[w]}G_C = y^{[w]}\in \F^{b\times n}$, and we write $y^{[w]} = \left(y^{[w]}_1,\ldots,y^{[w]}_n\right)$ with $y^{[w]}_j\in \F^{b\times 1}$.
	Let $D\subseteq \F^n$ be a linear code, and let 
	\[
	E=(e^{(1)},\dots , e^{(s)})\in (\F^{b\times n})^s
	\]
	be selected such that the composed map \begin{equation}\label{eq:retrieval}
	\F^{b\times k}\stackrel{\cdot G_C}{\longrightarrow} \F^{b\times n} \stackrel{\phi_E}{\longrightarrow} (\F^n)^s\stackrel{\cdot H}{\longrightarrow}(\F^{n-\dim(C\star D)})^s\end{equation} is injective, where $H$ is a parity check matrix of $C\star D$ and $\phi_E$ is  \begin{align*}\phi_E : \; \F^{b\times n}&\to (\F^n)^s\\
	y^{[w]}&\mapsto \left(\left(e^{(\gamma)}_1\cdot y^{[w]}_1,\dots, e^{(\gamma)}_n\cdot y^{[w]}_n \right) : 1\leq \gamma\leq s\right)\end{align*} 
	where $e^{(\gamma)}_j\in\F^{1\times b}$ denotes the $j^{th}$ row of $e^{(\gamma)}$.
The three maps in~\eqref{eq:retrieval} should be interpreted as encoding, receiving responses, and decoding, respectively.	

	The PIR scheme proceeds as follows:
	\begin{enumerate}
	\item  Select $Msb$ codewords independently and uniformly at random from $D$:
	\[
	d^{i(\gamma)\beta}\in D, \mbox{ for }1\leq i\leq M,\ 1\leq\gamma\leq s,\ 1\leq\beta\leq b.
	\] 
	
	\item  For $\gamma=1,\dots,s$, send the query 
		\[
		q^{[w](\gamma)}_j = (d^{1(\gamma)}_j,\ldots,d^{w(\gamma)}_j+ e^{(\gamma)}_j,\ldots,d^{M(\gamma)}_j)\in \F^{1\times M b}\] 
		to the $j^{th}$ server, where $d^{i(\gamma)}_j$ is the row vector $(d^{i(\gamma)1}(j),\dots , d^{i(\gamma)b}(j))\in \F^{1\times b}$ consisting of all of the $j^{th}$ entries of all of the vectors $d^{i(\gamma)\beta}$.
		
	\item Project the responses 
		\[
		\begin{split}\label{responses}
		\left(r^{[w](\gamma)}_1,\ldots,r^{[w](\gamma)}_n\right) &= \left(q^{[w](\gamma)}_1\cdot y_1,\ldots,q^{[w](\gamma)}_1\cdot y_1\right) \\
		&\in C\star D + \left(e^{(\gamma)}_1\cdot y^{[w]}_1,\dots, e^{(\gamma)}_n\cdot y^{[w]}_n \right).
		\end{split}
		\]
		to $(C\star D)^\perp$, via right-multiplying with the matrix $H$. \end{enumerate} We refer to this as a $(D,E)$-retrieval scheme. By injectivity of \eqref{eq:retrieval}, this scheme retrieves the file $x^{[w]}$ from $ns$ queries. It protects against the colluding set $T$ if $e^{(\gamma)}_T\in D_T$ for $\gamma=1,\dots s$.

	\begin{example}\label{ex:fullcoll}
	Suppose that $1\leq t\leq n -k$.  By choosing $C$ and $D$ to both be generalized Reed-Solomon (GRS) codes with the same evaluation vector, the $(D,E)$-retrieval scheme of \cite{FGHK16} can achieve a PIR rate of $\frac{n-(k+t-1)}{n}$ while protecting against $t$-collusion.  See \cite{FGHK16} for more details.
	\end{example}
	
	\subsection{Star Product Schemes for Non-MDS Coded Data}\label{sec:2B}
	To apply the above PIR scheme in our current setting, we need to generalize Theorem \ref{thmPIR} by removing the assumption on the set $J$, since such a set does not in general exist when $C$ and $D$ are not MDS codes. To prove the general theorem, we will need the following lemma:

	

	\begin{lemma}\label{lm:manyinfosets}
		An $[n,k,d_C]$-code $C$ has at least $\lceil d_C/k\rceil$ disjoint information sets.
			\end{lemma}
		\begin{proof}
Start with an arbitrary information set $S_1$. We will construct disjoint information sets $S_1,\dots, S_{\lceil d_C/k\rceil}$, each of size $k$, as follows. Inductively, for $1\leq i< d_C/k$, consider the code $C$ projected to the complement of $S_1\cup\cdots\cup S_i$. This projection has dimension $k$, since $C$ can correct $ik\leq d_C-1$ erasures. Thus, there will be an information set in the remaining coordinates, which we choose as $S_{i+1}$.
		\end{proof}

We can now prove the following generalization of Theorem \ref{thmPIR}.

	\begin{theorem}\label{thmPIR2}
	Let $C\subseteq\F^n$ be an $[n,k,d_C]$ linear storage code and let $D\subseteq \F^n$ be any linear code.  Then there exists a $(D,E)$-retrieval scheme for the distributed storage system $Y = XG_C$ with PIR rate $(d_{C\star D}-1)/n$ which protects against $(d_{D^\perp}-1)$-collusion.
	\end{theorem}
	\begin{proof}
Let $x^{[w]}$ denote the file we wish to download.  Set $c:=d_{C\star D}-1\leq d_C -1$.  Let us first suppose that $c \leq k$.  Choose an information set $S\subseteq [n]$ of $C$; after relabeling the storage nodes we may assume $S = \{1,\ldots,k\}$.   Let the files be spread over $b$ rows, i.e.\ $x^i \in \F^{b\times k}$ for all $i$, where $b = \lcm(c,k)/k$.  We use $s = \lcm(c,k)/c$ iterations of the scheme on the information set $S$ as follows.  Let $a=c/b$. During the $\gamma^{th}$ iteration of the PIR scheme, for each $\beta\in[b]$ we download (by putting a $1$ in the corresponding coordinate of $e^{(\gamma)}_j$) the $\beta^{th}$ entry from each of the vectors $y^{[w]}_{a(\gamma+\beta-2)+1}, y^{[w]}_{a(\gamma+\beta-2)+2}, \dots, y^{[w]}_{a(\gamma+\beta-1)}$, where all indices are computed modulo $k$.  One can check easily that after $s$ iterations, we have downloaded $k$ unique symbols from each row of the encoded file, which suffices to reconstruct the desired file since $S$ is an information set.  The resulting PIR rate is easily seen to be $\frac{bk}{ns} = c/n$, as desired.  See Fig.\ \ref{vis_scheme} for an illustration of which symbols are downloaded from $y^{[w]}$ during which iteration when $c = 4$ and $k = 6$.

	\begin{figure}[h]\centering{
			\begin{tikzpicture}
		\matrix [matrix of math nodes,left delimiter=(,right delimiter=),row sep=0.1cm,column sep=0.2cm, ampersand replacement=\&,nodes in empty cells] (m) {
1\&1\& 2\&2\& 3\& 3\& \cdots\\
3\&3\& 1\&1\& 2\&2\& \cdots
			\\};
		\node[fit=(m-1-1)(m-1-2),  draw, rounded corners,inner sep= 0pt] {};
		\node[fit=(m-2-3)(m-2-4),  draw, rounded corners,inner sep= 0pt] {};
		\node[fit=(m-1-3)(m-1-4),  draw, rounded corners,inner sep= 0pt, fill = gray!20, opacity=0.3] {};
		\node[fit=(m-2-5)(m-2-6),  draw, rounded corners,inner sep= 0pt, fill = gray!20, opacity=0.3] {};
		\node[above=10pt of m-1-1] (top-1) {$ y^{[w]}_1$};
		\node[above=10pt of m-1-2] (top-3) {$ y^{[w]}_{2}$};
		\node[above=10pt of m-1-3] (top-4) {$ y^{[w]}_{3}$};
		\node[above=10pt of m-1-4] (top-6) {$ y^{[w]}_{4}$};
		\node[above=10pt of m-1-5] (top-7) {$ y^{[w]}_{5}$};
		\node[above=10pt of m-1-6] (top-9) {$ y^{[w]}_{6}$};
		\node[above=10pt of m-1-7] (top-9) {$\cdots$};
		
		\end{tikzpicture}}
		\caption{An illustration of the scheme in the case $c = 4$, $k=6$. Each file consists of $b = 2$ rows and the scheme requires $s = 3$ iterations.  During a single iteration, $a = 2$ symbols are downloaded from each row.  An entry $\gamma$ in position $(\beta ,j)$ in this matrix means that in repetition $\gamma$ of the PIR protocol the symbol in column/server $j$ from the $\beta^{th}$ row is retrieved. }\label{vis_scheme}
	\end{figure}

Now suppose that $c \geq k$ and write $c = g\cdot k + \bar{c}$ with $0\leq \bar{c} < k$.  If $\bar{c} = 0$, then since $c = d_{C\star D}-1\leq d_C-1$, Lemma~\ref{lm:manyinfosets} guarantees that we can find $g = c/k$ disjoint information sets $S_1,\ldots,S_g$ of $C$.  In this case we can download the full $k$ information symbols from each of the $c/k$ information sets, using a different row of the file for each information set.  Thus when $\bar{c} = 0$ the scheme is essentially complete;  each file $x^i\in \F^{b\times k}$ requires $b = g$ rows and the scheme requires $s = 1$ iteration.

Lastly, suppose that $0<\bar{c} <k$.  Again by Lemma~\ref{lm:manyinfosets} there exist $g+1 = \lceil c/k \rceil$ disjoint information sets $S_1,\ldots,S_{g+1}$ of $C$.  Every iteration of the PIR scheme, we download $k$ symbols from each of $S_1,\ldots,S_g$ as in the previous paragraph, and an additional $\bar{c}<k$ symbols from the last information set $S_{g+1}$.  Letting $\bar{b} = \lcm(\bar{c},k)/k$ and $s = \lcm(\bar{c},k)/\bar{c}$ we download from $S_{g+1}$ as in the case $c<k$, using $\bar{b}$ rows of the desired file and $s$ repetitions of the scheme.  A file $x^i \in \F^{b\times k}$ now consists of $b = g\cdot s + \bar{b}$ rows, divided into $s$ rows for each information set $S_1,\ldots,S_g$ and $\bar{b}$ for the last information set $S_{g+1}$.   This completes the scheme construction.  The PIR rate is clearly $c/n$.

Now consider a set $T$ of size $t=d_{D^\perp}-1$. In each iteration, the query restricted to $T$ is a uniformly random element in $\F^{t\times bM}$, as the code $D$ has full rank on~$T$. Moreover, the sources of randomness in different iterations of the scheme are independent. Thus, the queries that $T$ observes throughout the course of the PIR scheme are uniformly random on $\left(\F^{t\times bM}\right)^s$, and hence do not depend on the desired file. This proves that the scheme protects against $t$-collusion.
 \end{proof}

\begin{remark}	
The construction contained in the above proof essentially gives a method for extending an incomplete PIR scheme which downloads $d_{C\star D}-1\leq d_C-1$ encoded symbols, one from each of a set of $d_{C\star D}-1$ servers, into a proper PIR scheme which can download a whole file with PIR rate $(d_{C\star D}-1)/n$.
\end{remark}

\subsection{Improving the PIR Rate of Star Product Schemes}
~\label{sec:improve}

The $(D,E)$-retrieval scheme of Theorem \ref{thmPIR2} essentially projects a vector fof weight $\leq d_{C\star D}-1$ onto the space $(C\star D)^\perp$, and takes advantage of the fact that any such vector can be recovered from this projection.  However, if we choose the vectors $e^{(\gamma)}_j$ more carefully, we can in principle recover some $\dim (C\star D)^\perp\geq d_{C\star D}-1$ coordinates of $y^{[w]}$.  Reed--Muller codes are in general not MDS, and hence this inequality will usually be strict, allowing us to increase the PIR rate of our retrieval scheme.

%

The following generalization of our Theorem \ref{thmPIR2} allows us to increase the PIR rate of $(D,E)$-retrieval schemes as described above, for certain choices of $C$ and $D$.  While the technical conditions to be checked to invoke the theorem are somewhat cumbersome, they are stated this way to somewhat axiomatize an approach to constructing many $(D,E)$-retrieval schemes with high rate.  

\begin{theorem}\label{thmPIR3}
Consider an $[n,k,d_C]$ linear storage code $C$ and any length $n$ linear code $D$. Assume that there exist (not necessarily distinct) subsets $S_1, \dots , S_b$ and $J_1,\dots , J_s$ of $[n]$ such that:\begin{itemize}
\item[(i)] $S_\beta$ is an information set of $C$ for $\beta=1,\dots , b$.
\item[(ii)] $J_\gamma$ is contained in an information set of $(C\star D)^\perp$ for $\gamma=1,\dots , s$.
\item[(iii)] For each $j\in [n]$ we have \begin{equation}\label{eq:uniformity}
\#\{\beta: j\in S_\beta\}=\#\{\gamma: j\in J_\gamma\}.
\end{equation}
\end{itemize} 
Then there exists a $(D,E)$-retrieval scheme for the distributed storage system $Y = XG_C$ with PIR rate $\frac{bk}{ns}$ which protects against $(d_{D^\perp}-1)$-collusion.  In particular, if all the sets $J_\gamma$ have the same cardinality $c$, then the PIR rate is $c/n$.
\end{theorem}
\begin{proof}
Recall that a file $x^{[w]}$ we wish to download has dimension $b\times k$.  We will download the $\beta^{th}$ row of $x^{[w]}$ via the information set $S_\beta$ for $1\leq \beta\leq b$, and in the $\gamma^{th}$ iteration, we will download one symbol from each column in $J_\gamma$ for $1\leq \gamma\leq s$. 

Fix an iteration $\gamma$, and inductively, assume we have defined the matrices $e^{(\gamma')}$ for all $\gamma'<\gamma$. 
For $j\in J_\gamma$, let $\beta_j\in [b]$ be the smallest index such that $j\in S_{\beta_j}$ and $e^{(\gamma')\beta_j}_j=0$ for all $\gamma'<\gamma$ (note that $\beta_j$ depends on $\gamma$). Now define the $n\times b$ matrix 
\[
e^{(\gamma)} = \left(e^{(\gamma)\beta}_j\right)_{\substack{1\leq j\leq n \\ 1\leq \beta\leq b}}\quad\text{by}\quad
e^{(\gamma)\beta}_j=\left\{\begin{tabular}{ll}
$1$ & \mbox{if } $j\in J_\gamma,\ \beta=\beta_j$ \\
$0$ & \mbox{otherwise}\\
\end{tabular}
\right. \] 
By~\eqref{eq:uniformity}, for $j\in[n]$ and $\beta\in [b]$, there is some $\gamma$ with $e^{(\gamma)\beta}_j=1$ if and only if $j\in S_\beta$.  For $j\in [n]$, let $e^{(\gamma)}_j\in \F^{1\times b}$ denote the $j^{th}$ row of $e^{(\gamma)}$.  If $j\in J_\gamma$, this is a standard basis vector with a $1$ in the $\beta_j^{th}$ position, and if $j\not\in J_\gamma$, this is the zero vector.

Recall that the encoded version of $x^{[w]}$ is denoted by $x^{[w]}G_C = y^{[w]}\in \F^{b\times n}$, and we write $y^{[w]} = \left(y^{[w]}_1,\ldots,y^{[w]}_n\right)$ with $y^{[w]}_j\in \F^b$.  During the $\gamma^{th}$ iteration, the relevant part of the total response as in (\ref{responses}) is of the form
\begin{equation}\label{complicated_response}
\left( e^{(\gamma)}_1\cdot y^{[w]}_1,\ldots,e^{(\gamma)}_n\cdot y^{[w]}_n\right)
\end{equation}
By construction, we have
\begin{equation}
e^{(\gamma)}_j\cdot y^{[w]}_j = \left\{\begin{array}{cl}
y^{[w]}_j(\beta_j) & j\in J_\gamma \\
0 & j\not\in J_\gamma
\end{array}\right.
\end{equation}
where $y^{[w]}_j(\beta_j)$ is the $\beta_j^{th}$ entry of $y^{[w]}_j$.  The assumption that $J_\gamma$ is contained in an information set of $(C\star D)^\perp$ implies that we can recover all of the $|J_\gamma|$ non-zero entries $\{y^{[w]}_j(\beta_j):j\in J_\gamma\}$ of the vector in (\ref{complicated_response}) after right-multiplication by the parity-check matrix of $C\star D$.  

So in the $\gamma^{th}$ iteration, we download $y^{[w]}_j(\beta_j)$ for each $j\in J_\gamma$, and after all iterations we have downloaded $y^{[w]}_j(\beta)$ for all $j\in S_\beta$ and all $\beta\in [b]$. Since $S_\beta$ is an information set of $C$, this allows us to recover the preimage $\left(x^{[w]}_1(\beta),\ldots,x^{[w]}_k(\beta)\right)\in\F^{1\times k}$, the $\beta^{th}$ row of $x^{[w]}$, for all $\beta$. Thus, the chain in \eqref{eq:retrieval} is injective, so $E=(e^{(1)},\dots e^{(s)})$ satisfies the criteria for the $(D,E)$-scheme to download the file $x^w$.  
The scheme is again easily seen to have rate $\frac{bk}{ns}$ and protect against $(d_{D^\perp}-1)$-collusion.  
\end{proof}


One can easily deduce Theorem \ref{thmPIR2} as a corollary of Theorem \ref{thmPIR3} by setting $c = d_{C\star D}-1\leq d_C-1$.  Indeed, by Lemma~\ref{lm:manyinfosets}, the code $C$ has at least $b'=\lceil c/k \rceil$ disjoint information sets $S_0,\dots , S_{b'-1}$. Let $b=cb'$, and let $S_{\beta} = S_{\beta \text{ (mod $b'$)}}$ for $0\leq \beta\leq b-1$ (for notational convenience, in this argument we index rows and iterations starting at $0$). After relabeling the servers, we can assume that $S_0\cup\cdots\cup S_{b'-1}=\{1,\dots , kb'\}$. Let $s=kb'$, and for $\gamma=0,\dots, s-1$, let $J_\gamma=\{\gamma+1, \dots , \gamma+c\}$, where in $J_\gamma$ a server index $j$ is to be understood as $j \text{ (mod $s$)} + 1$. Clearly, every element $j\in [s]$ is in precisely~$c$ sets $S_\beta$, and in precisely~$c$ sets~$J_\gamma$. As every set of size $c<d_{C\star D}$ is contained in an information set of $(C\star D)^\perp$, Theorem~\ref{thmPIR2} follows immediately.  

On the other hand, a fundamental upper bound on the download rate of a $(D,E)$ PIR scheme is $\dim(C\star D)^\perp/n$, as in each iteration we are downloading a projection to the space $(C\star D)^\perp$. So for every possible choice of $C$ and $D$, the maximal possible download rate of a $(D,E)$-PIR scheme lies between $(d_{C\star D} - 1)/n$ and $\dim(C\star D)^\perp/n$.  We will show in the coming sections that for a very significant class of codes, including Reed--Muller codes, we can always download at rate $\dim(C\star D)^\perp/n$.

\begin{corollary}\label{easy}
With $C$ and $D$ as in Theorem \ref{thmPIR3}, suppose there exists an information set $S$ of $C$ such that every subset of $S$ of size $\dim (C\star D)^\perp$ is an information set of $(C\star D)^\perp$.  Then we can achieve a PIR rate of $\dim(C\star D)^\perp/n$ while protecting against $(d_{D^\perp}-1)$-collusion.
\end{corollary}
\begin{proof}
We take the collection $J_1,\ldots,J_s$ to be all of the subsets of $S$ of size $c = \dim(C\star D)^\perp$, so that $s = \binom{k}{c}$.  It follows immediately that every $j\in S$ is contained in exactly $b = c\binom{k}{c}/k$ of the subsets $J_\gamma$.  Now define the sets $S_1,\ldots,S_b$ by simply setting $S_\beta = S$ for all $\beta$, with $1\leq \beta\leq b$.  The conditions of Theorem \ref{thmPIR3} are clearly satisfied, hence the result.
\end{proof}

In \cite{norwegians}, the authors study PIR for storage codes $C$ with code rate greater than $1/2$, with no server collusion.  We can apply Corollary \ref{easy} to some of the example codes they study, to obtain PIR rates which match or improve on the rates in \cite{norwegians}.

\begin{example}
Let $C$ be the $[5,3,2]_2$ code $C_1$ from \cite{norwegians}, defined by the parity-check matrix
\begin{equation}
G_{C^\perp} = {
\left(
\begin{array}{ccccc}
1 & 1 & 0 & 1 & 0 \\
0 & 1 & 1 & 0 & 1
\end{array}
\right)
}
\end{equation}
and let $S = \{1,2,3\}$ be the systematic information set of $C$.  Every $2$-subset of $S$ is clearly an information set of $C^\perp$, so Corollary \ref{easy} gives a PIR rate of $2/5$, the same as obtained in \cite{norwegians}.  Note that using Theorem \ref{thmPIR2}, we would have only achieved a PIR rate of $1/5$.

Similarly, let $C$ be the $[11,6,4]_2$ code $C_2$ from \cite{norwegians}, defined by the parity-check matrix
\begin{equation}
G_{C^\perp} = {
\left(
\begin{smallmatrix}
1 & 0 & 0 & 0 & 0 & 1 & 1 & 1 & 1 & 0 & 0 \\
0 & 1 & 0 & 0 & 0 & 1 & 0 & 0 & 0 & 0 & 1 \\
0 & 0 & 1 & 0 & 0 & 0 & 1 & 1 & 0 & 1 & 1 \\
0 & 0 & 0 & 1 & 0 & 0 & 1 & 0 & 1 & 1 & 0 \\
0 & 0 & 0 & 0 & 1 & 1 & 0 & 1 & 1 & 1 & 1
\end{smallmatrix}
\right)
}
\end{equation}
Using $G_{C^\perp}$, one checks that every $5$-subset of the information set $S = \{1,2,3,4,6,10\}$ of $C$ is an information set of $C^\perp$, yielding a PIR rate of $5/11$.  This improves on the PIR rate of $4/11$ achieved in \cite{norwegians} (itself an improvement over the PIR rate of $(d_C-1)/n = 3/11$ achieved by Theorem \ref{thmPIR2}).
\end{example}

As the next example shows, the requirement that every subset of $S$ is an information set of $(C\star D)^\perp$ is too strict, and sometimes we can achieve the same optimal rates by imposing less symmetry.  The following example illustrates this principle when $D = \rep(n)$.

\begin{example}
Consider the code $C_3$ of \cite{norwegians}, the $[12,8,4]_{11}$ Pyramid code from \cite{pyramids}.  With a generator matrix as in \cite[Section 2.2]{pyramids}, let $S$ be the information set $S = \{1,\ldots,8\}$ of $C$, and let $J$ be the information set $J = \{1,2,3,5\}$ of $C^\perp$.  Define the collection
\begin{equation}
\mathcal{J} = \{\{1,2,3,5\}, \{2,3,4,6\},\ldots,\{8,1,2,4\}\}
\end{equation}
of all ``cyclic shifts'' of $J$ within $S$.  One can check that every element of $\mathcal{J}$ is an information set of $C^\perp$.  An argument similar to the proof of Corollary \ref{easy} (we omit the details) shows that we can achieve a PIR rate of $4/12 = \dim(C^\perp)/n$, the same as obtained in \cite{norwegians}.

Similarly, consider the code $C = C_4$ of \cite{norwegians}, a $[16,10,5]_{16}$ locally repairable code from \cite{elephants}.  Defining $C$ via the generator matrix from \cite[Equation (7)]{elephants}, one can compute that the information set $S = \{1,\ldots,10\}$ of $C$ and the subset $J = \{1,\ldots,6\}$ have the property that every cyclic shift of $J$ within $S$ is an information set of $C^\perp$.  Again we achieve a PIR rate of $6/16=\dim(C^\perp)/n$, an improvement of the rate of $5/16$ obtained in \cite{norwegians} (itself an improvement over $4/16$, the rate obtained using Theorem \ref{thmPIR2}).
\end{example}

Note that in each of the four above example codes, we obtain a PIR rate of $\dim(C^\perp)/n$, the maximum possible for a $(\rep(n),E)$ retrieval scheme for the distributed storage system $Y = XG_C$.  In the next section, we show that if $C$ and $C\star D$ are transitive codes, we can always find a $(D,E)$ retrieval scheme which achieves the upper bound of $\dim(C\star D)^\perp/n$.

\subsection{PIR Schemes from Transitive Codes}\label{sec:trans}

We denote the symmetric group of $n$ elements by $\Sym_n$. For $c=(c_1,\dots c_n)\in \F^n$ and $\sigma\in \Sym_n$, define $\sigma(c)=(c_{\sigma(1)}, \dots,c_{\sigma(n)})\in \F^n$. This clearly defines a group action of $\Sym_n$ on $\F^n$. If $C$ is a linear code and $\sigma\in \Sym_n$ is such that $\sigma(c)\in C$ for every $c\in C$, then $\sigma$ is said to be a permutation automorphism, or simply automorphism, of $C$. The automorphisms of $C$ form a subgroup $\Gamma(C)\subseteq \Sym_n$. We note that every $\sigma \in \Gamma(C)$ maps information sets of $C$ into information sets. Moreover, note that for any code $C$ we have $\Gamma(C)=\Gamma(C^\perp)$. Recall that a subgroup $G\subseteq \Sym_n$ is {\em transitive} on $[n]$ if, for every $u,v\in [n]$, there exists $\sigma\in G$ with $\sigma(u)=v$.

\begin{lemma}\label{sets}
	Let $G$ and $H$ be any two transitive subgroups of $\Sym_n$, and let $S$ and $J$ be any two non-empty subsets of $[n]$.  Then there exist collections $\mathcal{S} = \{S_\beta\}$ and $\mathcal{J} = \{J_\gamma\}$ of (not necessarily distinct) subsets of $[n]$, such that
	\begin{itemize}
		\item[(i)] $S_\beta\in G\cdot S$ (the orbit of $G$ on $S$) for all $\beta$,
		\item[(ii)] $J_\gamma\in H\cdot J$ for all $\gamma$, and
		\item[(iii)] for all $j\in [n]$ we have $\#\{\beta: j\in S_\beta\}=\#\{\gamma: j\in J_\gamma\}$.
	\end{itemize} 
\begin{proof}
	Since $G$ is transitive the number of sets in the orbit $G\cdot S$ that contain a given element in $[n]$ is $x=\frac{|G||S|}{n}$, and hence independent of the chosen element. Analogously each element appears in $y=\frac{|H||J|}{n}$ sets of the orbit $H \cdot J$. Let $\alpha, \beta$ be chosen such that $\lcm(x,y)=\alpha x=\beta y$, then we see that the collection $\mathcal{S}$ containing $\alpha$ copies of the orbit $G \cdot S$ and the collection $\mathcal{J}$ containing $\beta$ copies of the orbit $H \cdot J$ both contain each element of $[n]$ in exactly $\lcm(x,y)$ of their sets.
\end{proof}
\end{lemma}
%

\begin{theorem}\label{thm:trans}
Let $C$ and $D$ be codes of length $n$ such that $\Gamma(C)$ and $\Gamma(C\star D)$ are transitive on $[n]$. Then there is a $(D,E)$-retrieval scheme for the distributed storage system $Y=XG_C$ with PIR rate $\dim(C\star D)^\perp/n$ which protects against $(d_{D^\perp}-1)$-collusion.
\end{theorem}
\begin{proof}
We apply Lemma \ref{sets} with $G = \Gamma(C)$, $H = \Gamma((C\star D)^\perp)$, $S$ an information set of $C$, and $J$ an information set of $(C\star D)^\perp$.  Let $\mathcal{S}$ and $\mathcal{J}$ be as in the lemma, and note that every $S_\beta\in \mathcal{S}$ is an information set of $C$, and every $J_\gamma\in \mathcal{J}$ is an information set of $(C\star D)^\perp$.  Thus, the collections of information sets $\mathcal{S}$ and $\mathcal{J}$ satisfy the assumptions of Theorem~\ref{thmPIR3}. As each of the sets $J_\gamma$ has cardinallity $\dim(C\star D)^\perp$, it follows that the PIR rate of the scheme in Theorem~\ref{thmPIR3} is $\dim(C\star D)^\perp/n$.
\end{proof}

Note that when using Theorem \ref{thm:trans} to construct a PIR scheme, the resulting number of rows per file and iterations can be calculated with the notation of Lemma \ref{sets} as $b=\alpha |G|$ and $s=\beta |H|$, respectively.

\section{Reed--Muller Codes for PIR}
\label{sec:3}

\subsection{Basic Definitions}

In this subsection we define and give some well-known results on Reed--Muller codes. We note that there are various ways to define Reed--Muller codes; for our purposes it is most convenient to view them as evaluation codes from multivariate polynomials. 

\begin{definition}
Let $0\leq r \leq m$ be integers and let $P_1,\dots,P_{2^m}$ be all the points of $\F_2^m$. Then the \emph{$r$-th order Reed--Muller code} of length $n=2^m$, denoted by $RM(r,m)$, is defined as
$$ RM(r,m) :=\Big\{ (f(P_1), \dots, f(P_n))  \mid  f \in \F_2[x_1,\dots,x_m], \deg f \leq r \Big\} .$$
\end{definition}

We need the following properties of Reed--Muller codes:

\begin{lemma}\label{lemRM}\cite[Ch. 13]{ma77}  Reed--Muller codes satisfy the following properties:
\begin{itemize}
\item[(i)] $RM(r,m)$ is a linear code of dimension $k=\sum_{i=0}^r \binom{m}{i}$.
\item[(ii)] $RM(r,m)$ has minimum distance $2^{m-r}$.
\item[(iii)] If $0\leq r<m$, then the dual code of $RM(r,m)$ is $RM(m-r-1,m)$.
\end{itemize}
\end{lemma}

To analyze the performance of a PIR scheme which uses Reed--Muller codes, we need to understand the star product of two such codes.
The following result is well-known, but for the sake of completeness we provide a short proof. 

\begin{lemma}\label{lemRM2} 
If $r+r'\leq m$, then 
$RM(r,m) \star RM(r',m) = RM(r+r', m). $
\end{lemma}
\begin{proof}
It is easy to see that $RM(r,m) \star RM(r',m)$ consists of the evaluation vectors of all $f \in \F_2[x_1,\dots,x_m]$ having degree less than or equal to $r+r'$. \end{proof}

It follows from Lemma \ref{lemRM} that $RM(r,m) \star RM(r',m)$ has minimum distance $2^{m-r-r'}$ and dimension $\sum_{i=0}^{r+r'} \binom{m}{i}$.

The number of minimal weight codewords of Reed--Muller codes is also known explicitly, and their structural description will be useful when proving quantitative bounds on the amount of collusion that our PIR schemes tolerate.

\begin{lemma}\label{lem:min_wt}\cite[Ch. 13, Thm. 8]{ma77}
Let $ c \in RM(r,m)$ be a codeword of minimal weight. Then $\supp ( c)\subseteq \F_2^m$ is an affine subspace of $\F_{2}^m$ of dimension $m-r$.
\end{lemma}
The next corollary follows by a simple counting argument.

\begin{corollary}\label{cor:min_wt}\cite[Ch. 13, Thm. 9]{ma77}
The number of minimum weight codewords in $RM(r,m)$ is\[2^{r}\frac{\prod_{i=0}^{m-r-1} (2^{m-i}-1)}{\prod_{i=0}^{m-r-1} (2^{m-r-i}-1)}.\]
\end{corollary}


\subsection{Achievable PIR Rate with Reed--Muller Codes}

We now choose $C=RM(r,m)$ as storage code, so that $n=2^m$ and $k=\sum_{i=0}^r \binom{m}{i}$. The code $D$ is chosen to be $RM(r',m)$ with $r+r'\leq m$, a code of the same length, but possibly different dimension.    Applying Lemma \ref{lemRM2} and Theorem \ref{thmPIR2}, we immediately obtain a $(D,E)$ retrieval scheme for $Y = XG_C$ with PIR rate $(d_{C\star D}-1)/n = (2^{m-(r+r')}-1)/2^m$ and which protects against $d_{D^\perp}-1 = 2^{m-r'}-1$ collusion.  However, as the following example illustrates, this na\"ive approach underestimates the achievable PIR rates when $C$ and $D$ are Reed--Muller codes.

\begin{example}
We consider $C=D=RM(1,4)$ with generator matrix
$$
 G_{C,D}= 
\left(\begin{smallmatrix}
1& 1& 1& 1& 1& 1& 1& 1& 1& 1& 1& 1& 1& 1& 1& 1\\  
1& 1& 1& 1& 1& 1& 1& 1& 0& 0& 0& 0& 0& 0& 0& 0\\  
1& 1& 1& 1& 0& 0& 0& 0 &  1& 1& 1& 1& 0& 0& 0& 0\\   
1& 1& 0& 0& 1& 1& 0& 0&  1& 1& 0& 0& 1& 1& 0& 0\\  
1& 0& 1& 0& 1& 0& 1& 0& 1& 0& 1& 0& 1& 0& 1& 0  
\end{smallmatrix}\right)\,,
$$
which is used to encode files $ x^i\in \F_2^{1\times 5}$.  We have $C\star D=RM(2,4)$, which has parity-check matrix $H_{C\star D}= G_{C,D}$ and minimum distance $2^{4-2}=4$. Hence we can achieve a PIR rate of $3/16$ and protect against $3$-collusion.

We can improve to a PIR rate of $\dim (C\star D)^\perp/n = 5/16$ as follows.  If we choose $e$ to be any vector of weight $5$ whose support corresponds to an invertible submatrix of $H_{C\star D}=G_{C,D}$, then we can download all of $ x^{[w]}$. For example, 
if we choose $e=(1110 1000  1000  0000)$, then a simple computation reveals that
\begin{equation}
(e\star  y^{[w]})H_{C\star D} =  x^{[w]} \left(\begin{smallmatrix}
1 & 0 & 0 & 0 & 0 \\
0 & 0 & 1 & 1 & 1 \\
0 & 1 & 0 & 1 & 1 \\
0 & 1 & 1 & 0 & 1 \\
0 & 1 & 1 & 1 & 0
\end{smallmatrix}\right)
\end{equation}
from which we can recover the whole file $ x^{[w]}\in \F_2^{1\times 5}$ by the invertibility of the above $5\times 5$ matrix.
\end{example} 

The previous example generalizes to the following result, which illustrates that with Reed--Muller codes, we can always achieve the upper bound $\dim(C\star D)^\perp/n$ for $(D,E)$ retrieval schemes.

%
\begin{corollary}\label{corollary:RMRM}
Let $C=RM(r,m)$ and let $0\leq r'<m-r$. Then there exists a $(D,E)$-retrieval scheme for the distributed storage system $Y = XG_C$ with a PIR rate of 
\begin{equation}
\frac{\sum_{i=0}^{m-r-r'} \binom{m}{i}}{2^m}\,,
\end{equation}
which protects against $(2^{r'+1}-1)$-collusion.
\end{corollary}
\begin{proof}
Let $D=RM(r',m)$, so that $C\star D=RM(r+r',m)$. By Theorem~\ref{thm:trans}, it is enough to show that $\Gamma(RM(r,m))$ and $\Gamma(RM(r+r',m))$ are transitive on the ground sets of the codes. Now note that the ground set of $RM(r,m)$ is $\F_2^m$, and affine transformations of $\F_2^m$ preserve the class of polynomials of degree $\leq r$. Thus, the affine transformations of $\F_2^m$ are automorphisms of $RM(r,m)$ and $RM(r+r',m)$. Since the affine transformations act transitvely on $\F_2^m$, so do the automorphism groups $\Gamma(RM(r,m))$ and $\Gamma(RM(r+r',m))$. By Theorem~\ref{thm:trans}, there is a $(D,E)$-retrieval scheme for the distributed storage system $Y=XG_C$ with PIR rate $\dim(C\star D)^\perp/n$ which protects against all colluding sets of size $d_{D^\perp}-1=2^{r'+1}-1$.  The result follows.
\end{proof}

In Figure \ref{fig:rates2} we use Corollary \ref{corollary:RMRM} to plot the resulting PIR rates of systems in which both $C$ and $D$ are Reed--Muller codes.  In the left-hand plot, we see the asymptotic behavior of the PIR rate for storage codes with code rate $1/2$, as the number of servers increases.  In the right-hand plot, we fix a system with $n = 64$ servers and observe how the tradeoff between the storage code rate and the PIR rate varies as we increase the amount of server collusion.

\begin{figure}[h]\centering	
	\begin{minipage}{.45\textwidth}
	\begin{tikzpicture}[scale=0.9]
	\begin{axis}[
	height=7cm,
	width=8cm,
	grid=major,
	xlabel={Length of code $n$},
	ylabel={PIR rate},
	ytick={0, 0.2, 0.4, 0.6},
	xmode=log,
	ymin=0,
	ymax=0.8,
	]
		\addplot+[blue,solid,mark=*] coordinates{%
		(8, 0.5)
		(32, 0.5)
		(128, 0.5)
		(512, 0.5)
		}; \addlegendentry{no collusion}
	\addplot+[red,solid,mark=square*] coordinates{%
		(8, 0.125)
		(32, 0.1875)
		(128, 0.2265625)
		(512, 0.25390625)
	}; \addlegendentry{$3$-collusion}
	\addplot+[brown,solid,mark=diamond*] coordinates{%
		(32, 0.03125)
		(128, 0.0625)
		(512, 0.08984375)
	}; \addlegendentry{$7$-collusion}
	\addplot+[blue,dashed,mark=*] coordinates{%
		(8, 0.5)
		(32, 0.5)
		(128, 0.5)
		(512, 0.5)
		};
	\addplot+[red,dashed,mark=square*] coordinates{%
		(8, 0.25)
		(32, 0.4375)
		(128, 0.484375)
		(512, 0.498046875)
	};
	\addplot+[brown,dashed,mark=diamond*] coordinates{%
		(32, 0.3125)
		(128, 0.453125)
		(512, 0.48828125)
	};
	\end{axis}
	\end{tikzpicture}
	\end{minipage} \hspace{.5cm}
	\begin{minipage}{.45\textwidth}
	\begin{tikzpicture}[scale=0.9]
        \begin{axis}[
            height=7cm,
            width=8cm,
            grid=major,
            xlabel={Code rate},
            ylabel={PIR rate},
            ytick={0, 0.2, 0.4, 0.6, 0.8, 1},
            xmin=0,
            ymin=0,
        ]
                     \addplot+[blue,solid,mark=*] coordinates{%
                (0.015625, 0.984375)
                (0.109375, 0.890625)
                (0.34375, 0.65625)
                (0.65625, 0.34375)
                (0.890625, 0.109375)
                (0.984375, 0.015625)
        }; \addlegendentry{no collusion}
        \addplot+[red,solid,mark=square*] coordinates{%
                (0.015625, 0.890625)
                (0.109375, 0.65625)
                (0.34375, 0.34375)
                (0.65625, 0.109375)
                (0.890625, 0.015625)
        }; \addlegendentry{$3$-collusion}
        \addplot+[brown,solid,mark=diamond*] coordinates{%
                (0.015625, 0.65625)
                (0.109375, 0.34375)
                (0.34375, 0.109375)
                (0.65625, 0.015625)
        }; \addlegendentry{$7$-collusion}
        \addplot+[black,solid,mark=asterisk] coordinates{%
                 (0.015625, 0.34375)
                (0.109375, 0.109375)
                (0.34375, 0.015625)
        }; \addlegendentry{$15$-collusion}
             \addplot+[olive,solid,mark=triangle*] coordinates{%
                 (0.015625, 0.109375)
                (0.109375, 0.015625)
        }; \addlegendentry{$31$-collusion}
                             \addplot+[blue,dashed,mark=*] coordinates{%
                (0.015625, 0.984375)
                (0.109375, 0.890625)
                (0.34375, 0.65625)
                (0.65625, 0.34375)
                (0.890625, 0.109375)
                (0.984375, 0.015625)
        };
        \addplot+[red,dashed,mark=square*] coordinates{%
                (0.015625, 0.953125)
                (0.953125, 0.015625)
        }; 
        \addplot+[brown,dashed,mark=diamond*] coordinates{%
                (0.015625, 0.890625)
                (0.890625, 0.015625)
        };
        \addplot+[black,dashed,mark=asterisk] coordinates{%
                 (0.015625, 0.765625)
                (0.765625, 0.015625)
        };
             \addplot+[olive,dashed,mark=triangle*] coordinates{%
                 (0.015625, 0.515625)
                (0.515625, 0.015625)
        }; 
        	 \end{axis}
	\end{tikzpicture}
	\end{minipage}
\caption{On the left, PIR rate for binary RM storage codes (solid) and GRS storage codes (dashed with corresponding color) of fixed code rate $1/2$. Note that the PIR rates agree in the case of no collusion. On the right, PIR rate versus storage code rate for binary RM storage codes (solid) and GRS storage codes (dashed) of fixed length $64$.}\label{fig:rates2}
\end{figure}
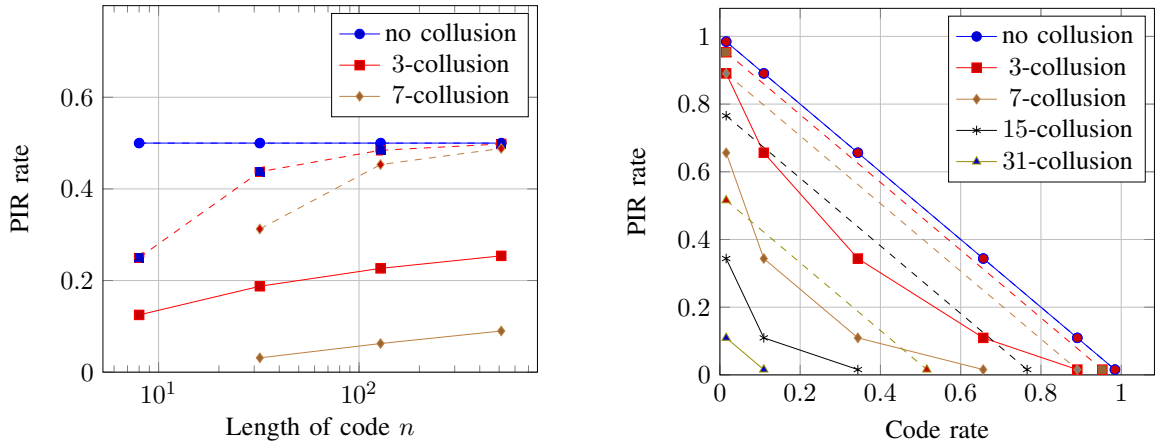

\begin{example}\label{ex:rm04}
Suppose that $C = RM(0,4) = \rep(16)_2$, so that data is stored via a replication system over $n = 16$ servers.  Set $D = RM(1,4)$, which is a $[16,5,8]_2$-code.  Then $(C\star D)^\perp = D^\perp = RM(2,4)$, which is a $[16,11,4]_2$-code.  The PIR scheme of Theorem \ref{thmPIR2} achieves a PIR rate of $7/16$ and protects against all colluding sets of size $2^{1+1}-1 = 3$.  With the scheme of Corollary \ref{corollary:RMRM}, we have a $(D,E)$-retrieval scheme with PIR rate $11/16$, which is a substantial improvement.

\end{example}
\begin{remark}
Using the scheme of \cite{FGHK16} with  $C = \rep(16)_q$ and $D$ a $[16,3,14]_q$-GRS code, one can protect against $3$-collusion while achieving a PIR rate of $(n-t)/n = 13/16$.  If we on the other hand fix the PIR rate to be $11/16$ and compare the privacy properties of the schemes in \cite{FGHK16} and the above example, the scheme of \cite{FGHK16} achieves this rate by keeping $C = \rep(16)_q$ and setting $D$ to be a GRS code with parameters $[16,5,12]_q$.  Then $D^\perp$ is GRS with parameters $[16,11,6]_q$, and hence we protect against $5$-collusion.  These improvements in either PIR rate or privacy require a field size of $q\geq 16$, while the Reed--Muller scheme is defined over $\F_2$.

On a more general note, a binary PIR scheme can also be set up with a GRS code over $\F_{2^h}$ for some integer $h>1$, where every symbol from  $\F_{2^h}$ is represented as an element in $\F_{2}^h$. However, one can easily check that the performance of these codes in terms of protection against colluding sets and PIR rate is poor. For example, consider $C=\rep(16)_2$ as above and $D$ the binary expansion of an $[4,1,4]$-GRS code over $\F_4$, which is a $[16,4,4]$-code over $\F_2$. This scheme has a PIR rate of $3/16$ and only protects against $1$-collusion. The Reed--Muller PIR scheme which sets $D = RM(1,4)$ is clearly preferable to this one.

\end{remark}


Using Corollary \ref{corollary:RMRM} to construct PIR schemes for arbitrary Reed--Muller codes involves computing the orbits of information sets of $C$ and $(C\star D)^\perp$ under the respective automorphism groups, and when applied directly may result in schemes which demand large numbers of rows per file and iterations.  However, as the following example shows, one can sometimes achieve the same rates with arguments similar to Corollary \ref{easy} and the subsequent examples.

\begin{example}
Consider a system with no server collusion and let $C = RM(2,4)$, which is defined by the parity-check matrix
\begin{equation}
G_{C^\perp} = 
\left(\begin{smallmatrix}
1 & 0 & 0 & 1 & 0 & 1 & 1 & 0 & 0 & 1 & 1 & 0 & 1 & 0 & 0 & 1 \\
0 & 1 & 0 & 1 & 0 & 1 & 0 & 1 & 0 & 1 & 0 & 1 & 0 & 1 & 0 & 1 \\
0 & 0 & 1 & 1 & 0 & 0 & 1 & 1 & 0 & 0 & 1 & 1 & 0 & 0 & 1 & 1 \\
0 & 0 & 0 & 0 & 1 & 1 & 1 & 1 & 0 & 0 & 0 & 0 & 1 & 1 & 1 & 1 \\
0 & 0 & 0 & 0 & 0 & 0 & 0 & 0 & 1 & 1 & 1 & 1 & 1 & 1 & 1 & 1
\end{smallmatrix}\right)\,.
\end{equation}
By Corollary~\ref{corollary:RMRM}, we can achieve a PIR rate of $5/16 = \dim(C^\perp)/n$, when setting $b=|J||\Aff(\F_2^4) J|$ and $s=|S||\Aff(\F_2^4) S|$, where $S$ and $J$ are information sets of $C=RM(2,4)$ and $C^\perp=RM(1,4)$ respectively, and $\Aff$ denotes the affine group. These information sets have size $|S|=11$ and $|J|=5$, and straightforward calculations show that they are stabilized by subgroups of order $5!=120$ of $\Aff(\F_2^4)$, which has order $322560$. Thus, a na\"ive application of  Corollary~\ref{corollary:RMRM} would require $b=5\cdot 322560/120=13440$ blocks per file and $s=11\cdot 322560/120=29568$ iterations.
 
Now, let $S$ be the information set $S = \{1, 2, 3, 4, 5, 6, 7, 9, 10, 11, 16\}$ of $C$ and consider the subset $J = \{1, 2, 3, 5, 10\}$.  Let $\mathcal{J}$ be the collection of all cyclic shifts of $J$ within $S$.  One checks using $G_{C^\perp}$ that every subset of $\mathcal{J}$ is an information set of $C^\perp$.  Hence we can achieve a PIR rate of $5/16 = \dim(C^\perp)/n$, when setting $b=5$ and $s=11$. This shows that a careful analysis of the information sets of the storage and retrieval codes can significantly improve the practicality of our schemes.
\end{example}

\subsection{Protection against $t$-Collusion}

When $t\geq d_{D^\perp}$, the Reed--Muller PIR scheme does not protect against all $t$-colluding sets of servers. However, for $t\approx d_{D^\perp}$, it does protect against ``most'' $t$-colluding sets in the following sense.     

\begin{proposition}\label{prop:coll}
Let $D=RM(r,m)$, and let \[d_{D^\perp}=2^{r+1}\leq t\leq \sum_{i=0}^{r}\binom{m}{i}=\dim(D).\] Let $T\subseteq \F_2^m$ be a set of $|T|=t$ servers, chosen uniformly at random. Then the probability that the PIR scheme does not protect against collusion in $T$ is bounded from above by \[\frac{\binom{2^m-2^{r+1}}{t-2^{r+1}}}{\binom{2^m}{t}}2^{m-r-1} \frac{\prod_{i=0}^{r} (2^{m-i}-1)}{\prod_{i=0}^{r} (2^{r+1-i}-1)}.\] If $t<3\cdot 2^r$, then this bound is tight.
\end{proposition}

\begin{proof}
We fail to protect against a colluding set $T$ if and only if $\dim(D|_T)< |T|$.  This latter condition is equivalent to the existence of a codeword of $D^\perp$ whose support is contained in $T$. 

By Corollary \ref{cor:min_wt}, there are $2^{m-r-1} \frac{\prod_{i=0}^{r} (2^{m-i}-1)}{\prod_{i=0}^{r} (2^{r+1-i}-1)}$ minimal length codewords in $D^\perp$. Each of these minimal codewords has its support contained in exactly $\binom{2^m-2^{r+1}}{t-2^{r+1}}$ sets of size $t$, so there exist at most \begin{equation}\label{eq:setcount}
\binom{2^m-2^{r+1}}{t-2^{r+1}}2^{m-r-1} \frac{\prod_{i=0}^{r} (2^{m-i}-1)}{\prod_{i=0}^{r} (2^{r+1-i}-1)}
\end{equation} $t$-sets that contain the support of some codeword in $D^\perp$. 

For the second statement, notice that by Lemma~\ref{lem:min_wt}, the support of two minimum weight codewords of $D^\perp$ intersect in a flat of dimension at most $r$ in $\F_2^m$. Thus, their union has size at least $2\cdot 2^{r+1}-2^r=3\cdot 2^r$. As a consequence, if $t<3\cdot 2^r$, then the collections of non-protected sets corresponding to different minimal codewords in $D^\perp$ are disjoint. Thus, the number of such sets is exactly given by~\eqref{eq:setcount}.
\end{proof}

\begin{example} Continuing Example~\ref{ex:rm04} wherein $D = RM(1,4)$, the $4$-colluding sets $T$ that we fail to protect against are in bijection with minimal weight codewords of $D^\perp$. By Corollary \ref{cor:min_wt} there are $120$ minimal weight codewords of $D^\perp$.  Hence the Reed--Muller PIR scheme protects against collusion for
\begin{equation}
\left(\tbinom{16}{4}-120\right)\tbinom{16}{4}^{-1}\approx 93.4\%
\end{equation}
of subsets of servers of size $t = 4$.

Similarly, there are $\binom{16}{5} = 4368$ subsets $T$ of servers of size $5$, of which $2688$ satisfy $\dim(D|_T) = 5$, according to Proposition~\ref{prop:coll}.  It follows that the scheme protects against collusion for $\frac{2688}{4368}\approx 61.5\%$ of all subsets of servers of size $5$.
\end{example}

\section{Conclusion}
\label{sec:conclusions}
In this paper we have studied PIR schemes for coded storage systems with colluding servers.  Given an arbitrary storage code $C$ and retrieval code $D$, we have constructed a PIR scheme with rate $(d_{C\star D}-1)/n$ which protects against $(d_{D^\perp}-1)$-collusion, where $n$ is the length of the codes as well as the number of servers in the system.  For some classes of $C$ and $D$, in particular when $C$ and $C\star D$ have transitive automorphism groups, we have shown that we can improve our scheme to have rate $\dim(C\star D)^\perp/n$, the maximum possible for the presented scheme.  In particular, this applies when $C$ and $D$ are binary Reed--Muller codes, resulting in a large class of PIR schemes defined over $\F_2$ for coded storage systems with colluding servers.  As a corollary of our results, we have improved on the PIR rates of some of the distributed storage systems studied in \cite{norwegians}. The rate $\dim(C\star D)^\perp/n$ also coincides with the asymptotic PIR capacity in the non-colluding case ($t=1$), the query code  $D$ then being a repetition code. 

Future work will consist of studying other important classes of transitive codes, as well as quantifying achievable PIR rates in terms of the automorphism groups of $C$ and $C\star D$.  In this work we have focused more on concrete constructions, but understanding the PIR capacity for various models when we limit the field size is a question worth pursuing.  Lastly, given a transitive code $C$, we plan on studying natural conditions on another code $D$ such that $C\star D$ is also transitive.  This may help apply our results to other meaningful classes of codes, such as general evaluation codes and locally repairable codes.  Cyclic codes provide an especially encouraging avenue of research, as they are transitive and the class of cyclic codes is closed under the star product.

\bibliographystyle{IEEEtran}
\bibliography{references}

\begin{thebibliography}{10}
\providecommand{\url}[1]{#1}
\csname url@samestyle\endcsname
\providecommand{\newblock}{\relax}
\providecommand{\bibinfo}[2]{#2}
\providecommand{\BIBentrySTDinterwordspacing}{\spaceskip=0pt\relax}
\providecommand{\BIBentryALTinterwordstretchfactor}{4}
\providecommand{\BIBentryALTinterwordspacing}{\spaceskip=\fontdimen2\font plus
\BIBentryALTinterwordstretchfactor\fontdimen3\font minus
  \fontdimen4\font\relax}
\providecommand{\BIBforeignlanguage}[2]{{%
\expandafter\ifx\csname l@#1\endcsname\relax
\typeout{** WARNING: IEEEtran.bst: No hyphenation pattern has been}%
\typeout{** loaded for the language `#1'. Using the pattern for}%
\typeout{** the default language instead.}%
\else
\language=\csname l@#1\endcsname
\fi
#2}}
\providecommand{\BIBdecl}{\relax}
\BIBdecl

\bibitem{chor1995private}
B.~Chor, O.~Goldreich, E.~Kushilevitz, and M.~Sudan, ``Private information
  retrieval,'' in \emph{IEEE Annual Symposium on Foundations of Computer
  Science}, 1995, pp. 41--50.

\bibitem{chor1998private}
B.~Chor, E.~Kushilevitz, O.~Goldreich, and M.~Sudan, ``Private information
  retrieval,'' \emph{Journal of the ACM (JACM)}, vol.~45, no.~6, pp. 965--981,
  1998.

\bibitem{efremenko20093}
K.~Efremenko, ``3-query locally repairable codes of subexponential length,'' in
  \emph{ACM Symposium on the Theory of Computing}, 2009, pp. 39--44.

\bibitem{dvir20142}
\BIBentryALTinterwordspacing
Z.~Dvir and S.~Gopi, ``{2-Server PIR with Sub-Polynomial Communication},'' in
  \emph{ACM Symposium on Theory of Computing}, 2015, pp. 577--584. [Online].
  Available: \url{http://doi.acm.org/10.1145/2746539.2746546}
\BIBentrySTDinterwordspacing

\bibitem{shah2012}
\BIBentryALTinterwordspacing
N.~B. Shah, K.~V. Rashmi, K.~Ramchandran, and P.~V. Kumar, ``Privacy-preserving
  and secure distributed storage codes,'' 2012. [Online]. Available:
  \url{http://people.eecs.berkeley.edu/~nihar/publications/privacy_security.pdf}
\BIBentrySTDinterwordspacing

\bibitem{shah2014}
N.~B. Shah, K.~V. Rashmi, and K.~Ramchandran, ``One extra bit of download
  ensures perfectly private information retrieval,'' in \emph{2014 IEEE
  International Symposium on Information Theory}, June 2014, pp. 856--890.

\bibitem{blackburn2016small}
S.~R. Blackburn, T.~Etzion, and M.~B. Paterson, ``{PIR} schemes with small
  download complexity and low storage requirements,'' in \emph{2017 IEEE
  International Symposium on Information Theory (ISIT)}, June 2017, pp.
  146--150.

\bibitem{razan_salim}
R.~Tajeddine and S.~{El Rouayheb}, ``Private information retrieval from {MDS}
  coded data in distributed storage systems,'' in \emph{2016 IEEE International
  Symposium on Information Theory (ISIT)}, July 2016, pp. 1411--1415, see
  {arXiv:1602.01458} for an extended version.

\bibitem{FGHK16}
R.~Freij-Hollanti, O.~W. Gnilke, C.~Hollanti, and D.~A. Karpuk, ``Private
  information retrieval from coded databases with colluding servers,''
  \emph{SIAM Journal on Applied Algebra and Geometry}, pp. 647--664, 2017.

\bibitem{patternISIT2017}
R.~Tajeddine, O.~W. Gnilke, D.~Karpuk, R.~Freij-Hollanti, C.~Hollanti, and
  S.~E. Rouayheb, ``Private information retrieval schemes for coded data with
  arbitrary collusion patterns,'' in \emph{2017 IEEE International Symposium on
  Information Theory (ISIT)}, June 2017, pp. 1908--1912.

\bibitem{sun_jafar_1}
H.~Sun and S.~A. Jafar, ``The capacity of private information retrieval,''
  \emph{IEEE Transactions on Information Theory}, vol.~63, no.~7, pp.
  4075--4088, July 2017.

\bibitem{sun_jafar_2}
------, ``The capacity of robust private information retrieval with colluding
  databases,'' \emph{IEEE Transactions on Information Theory}, 2017.

\bibitem{bananaman}
K.~Banawan and S.~Ulukus, ``The capacity of private information retrieval from
  coded databases,'' 2016, arXiv: 1609.08138.

\bibitem{sun_jafar_conjecture}
H.~Sun and S.~A. Jafar, ``Private information retrieval from {MDS} coded data
  with colluding servers: Settling a conjecture by {Freij-Hollanti} et al,'' in
  \emph{2017 IEEE International Symposium on Information Theory (ISIT)}, June
  2017, pp. 1893--1897.

\bibitem{norwegians}
S.~Kumar, E.~Rosnes, and {A. Graell i Amat}, ``Private information retrieval in
  distributed storage systems using an arbitrary linear codes,'' in \emph{2017
  IEEE International Symposium on Information Theory (ISIT)}, June 2017, pp.
  1421--1425.

\bibitem{vajha2017}
M.~Vajha, V.~Ramkumar, and P.~V. Kumar, ``Binary, shortened projective {R}eed
  {M}uller codes for coded private information retrieval,'' in \emph{2017 IEEE
  International Symposium on Information Theory (ISIT)}, 2017, pp. 2648--2652.

\bibitem{favaya2015}
A.~Fazeli, A.~Vardy, and E.~Yaakobi, ``Codes for distributed {PIR} with low
  storage overhead,'' in \emph{2015 IEEE International Symposium on Information
  Theory (ISIT)}, June 2015, pp. 2852--2856.

\bibitem{pyramids}
C.~Huang, M.~Chen, and J.~Li, ``Pyramid codes: Flexible schemes to trade space
  for access efficiency in reliable data storage systems,'' in \emph{IEEE
  International Symposium on Network Computing and Applications (NCA)}, 2007,
  pp. 79--86.

\bibitem{elephants}
M.~Sathiamoorthy, M.~Asteris, D.~Papailiopoulos, A.~G. Dimakis, R.~Vadali,
  S.~Chen, and D.~Borthakur, ``{XOR}ing elephants: Novel erasure codes for big
  data,'' in \emph{Proc. 39th Very Large Data Bases Endowment}, 2013, pp.
  325--336.

\bibitem{ma77}
F.~J. MacWilliams and N.~J.~A. Sloane, \emph{The Theory of Error-Correcting
  Codes}.\hskip 1em plus 0.5em minus 0.4em\relax Amsterdam: North Holland,
  1977.

\end{thebibliography}

\end{document}